\documentclass[12pt]{article}
\newtheorem{theorem}{Theorem}
\newtheorem{remark}{Remark}

\usepackage[utf8]{inputenc}
\usepackage[english]{babel}
\usepackage[centertags]{amsmath}
\usepackage{amssymb}
\usepackage{bm}
\usepackage{bbm}
\usepackage{amsfonts} 
\usepackage{amssymb}
\usepackage{amsbsy}
\usepackage{graphicx}
\usepackage{appendix}
\usepackage{mathrsfs}
\usepackage{epsfig}
\usepackage{color}
\usepackage{euscript}
\usepackage{ulem}
\usepackage{multirow}
\usepackage{cite}

\definecolor{dgreen}{rgb}{0,0.6,0}

\definecolor{darkblue}{rgb}{0., 0, 1}

\definecolor{purple}{rgb}{0.65,0.,0.78}

\definecolor{orange}{rgb}{0.89,0.42,0.05}

\usepackage{jheppubm}

\newcommand{\be}{\begin{equation}}
\newcommand{\ee}{\end{equation}}
\newcommand{\bea}{\begin{eqnarray}}
\newcommand{\eea}{\end{eqnarray}}

\newcommand{\cL}{{\cal L}}

\newcommand{\cK}{{\cal K}}

\numberwithin{equation}{section}

\title{Scattering and Tunneling in Real  Quantum Mechanics }

\author{Firdevs Karakus$^a$, Daniil Stepanenko$^b$ and Igor Volovich$^b$ }

\affiliation{
$^a$Department of Physics, Boğaziçi University, \\
34342 Bebek, Istanbul, Türkiye \\
$^b$Steklov Mathematical Institute, Russian Academy of  Sciences, \\ Gubkina str. 8, 119991, Moscow, Russia

}

\emailAdd{firdevs.karakus@std.boun.edu.tr, dstepanenko@mi-ras.ru, volovich@mi-ras.ru}

\abstract{We establish the general form of the evolution equation and study scattering and tunneling in real quantum mechanics. We prove an analogue of Stone's theorem for a real Kähler space and derive the corresponding general evolution equation, whose evolution operator is both orthogonal and symplectic. We formulate scattering theory in terms of real wave operators and the associated 
S-matrix. For the class of potentials considered, we show that the differential scattering cross section is equivalent to that obtained in standard complex quantum mechanics. We also show that the tunneling probability in real quantum mechanics is identical to that in complex quantum mechanics.
}

\begin{document}
\newpage
%%%%%%%%

%%%%%%%%
\newpage
\maketitle
%%%%%%%%%%%%%%%%%%%%%%%%%%%%%%%%%%%%%%%%%%%%%%%%%%%%%%BOLD BOLD BOLD BOLD BOLD BOLDBOLD
 %\bf\boldmath\Large
%%%%%%%%%%%%%%%%%%%%%%%%%%%%%%%%%%%%%%%%%%%%%%%%%%%%%%BOLD BOLD BOLD BOLD BOLD BOLDBOLD
\section{Introduction}
%%%%%%%%%%%%%%%%%%%%%%%%%%%%%%%%%%%%%%%%%%%%%%%%%%%%%%

In this paper, we consider the general form of the evolution equation and scattering theory in real quantum mechanics, with quantum tunneling as an illustrative example.

Standard quantum mechanics is formulated over the complex numbers on a complex Hilbert space. The possibility of a purely real formulation has been debated since the early days of the theory \cite{BirkhoffVonNeumann:1936logic,Varadarajan:1985gqt,vonNeumann:2018mfqm}. In our previous work \cite{Volovich:2025rmi,Arefeva:2025zbx}, we demonstrated that quantum mechanics admits a real formulation based on a real Kähler space—that is, a real Hilbert space equipped with a symplectic structure. For various discussions on whether real formulations reproduce the predictions of standard complex quantum mechanics, see
\cite{Stueckelberg:1960ejx,Strocchi66,Kozlov,VilelaMendes:2002ejy,Strocchi,Moretti:2016csf,Renou:2021dvp,Zhu:2020iml,Hita:2025okv,Volovich:2025rmi,Ying:2025xyl,Arefeva:2025zbx,Bang:2026lbt,Maioli:2026mhy,Chen:2021ril,Li:2021uof,2103.12740,Bednorz:2022zgq,Vedral:2023pij,Takatsuka:2025hez,Sarkar:2025tmd}.

In complex quantum mechanics, the general form of the evolution equation is given by the Schr\"odinger equation
\begin{equation}
i\dot{\psi}=H\psi,
\end{equation}
where $H$ is a self-adjoint operator on a separable Hilbert space. Stone's theorem relates such $H$ as the generator of a strongly continuous one-parameter unitary group. The resulting unitary time evolution preserves probabilities. In this paper, we prove an analogue of Stone's theorem for real quantum mechanics and obtain the general form of the evolution equation
\begin{equation}
\label{GEE}
\mathcal{J}\dot{\phi}=\mathcal{H}\phi(t).
\end{equation}
Here, $\phi(t)\in\mathcal{K}$, where $\mathcal{K}$ is a real K\"ahler space; $\mathcal{J}$ is the $2\times 2$ symplectic matrix satisfying $\mathcal{J}^{2}=-\mathbb I$; and $\mathcal{H}$ is a self-adjoint operator that commutes with $\mathcal{J}$
\begin{equation}
    [ \mathcal{H,J} ] = 0.
\end{equation}
Equation \eqref{GEE} is the real analogue of the Schr\"odinger equation on a complex Hilbert space. Since the corresponding evolution operator preserves both the inner product and the symplectic form on the real K\"ahler space, it is both orthogonal and symplectic.

The real quantum mechanics can be obtained from the usual complex formulation by representing multiplication by $i$ through the action of $\mathcal{J}$. In particular, in units where $\hbar=1$, the Heisenberg commutation relation is replaced by the $\mathcal{J}$-quantization rule 
\begin{equation}
[Q,P]=\mathcal{J}.
\end{equation}
For earlier discussions of time evolution in real quantum mechanics and real quantum field theory, see \cite{Stueckelberg:1960ejx,Strocchi66,VilelaMendes:2002ejy,Volovich:2025rmi,Arefeva:2026xlr,Passon:2026ags,Singh:2026dtp,Takatsuka:2025hez,Sarkar:2025tmd}.

In this paper, we formulate scattering theory in real quantum mechanics using the general evolution equation derived above. The wave operators in the real quantum mechanics are defined as
% \begin{equation}
% \Omega_\pm^{\mathbb{R}}
% =
% s-\lim_{t\to\pm\infty}
% e^{-t\mathcal{J}\mathcal{H}0}
% e^{t\mathcal{J}\mathcal{H}},
% \end{equation}
\be
\Omega_\pm^{\mathbb{R}}
=\mathop{\mathrm{s\text{-}lim}}_{t\to\pm\infty}e^{-t\mathcal{J}\mathcal{H}0}
e^{t\mathcal{J}\mathcal{H}},\ee
where $\mathcal{H}$ and $\mathcal H_0$ is a pair of self adjoint operators. We adapt Cook's method from scattering theory on complex Hilbert spaces to the real setting and prove the existence of the wave operators under natural assumptions. Under the small-Rollnik hypothesis, we further establish their asymptotic completeness. The associated scattering operator is then given by
\begin{equation}
S_{\mathbb{R}}
=
\left(\Omega_+^{\mathbb{R}}\right)^{*}
\Omega_-^{\mathbb{R}}.
\end{equation}
For central potentials, we show that the scattering cross section obtained in real quantum mechanics coincides with the corresponding cross section in standard complex quantum mechanics. 

In standard quantum mechanics, consideration of tunneling requires complex variables. We show that the real formulation, expressed entirely in terms of real variables, yields the same tunneling probability as standard complex quantum mechanics. Thus, symplectic dynamics on the real K\"ahler space reproduces quantum transmission through a classically forbidden barrier.

Note that there exist also another motivation for real quantum theory. it was proved \cite{IVpadic} that any quantum dynamical system is integrable and equivalent to a system of classical harmonic oscillators, which is the real dynamical system. This material is presented in the appendix.
 
The paper is organized as follows. Section 2 introduces the realification of the Schr\"odinger equation and discusses the finite dimensional case. Section 3 presents the spectral theorem for a real K\"ahler space. Section 4 establishes a real analogue of Stone's theorem and derives the general form of the evolution equation. Section 5 studies the extension of the wave function into a classically forbidden region and compares the result with standard complex quantum mechanics. Section 6 analyzes quantum tunneling in the real K\"ahler formulation. Section 7 develops scattering theory, establishes the existence and asymptotic completeness of the real wave operators under the stated assumptions, and shows that, for a class of potentials considered, the differential cross sections is equivalent with those obtained in standard complex quantum mechanics.

\section{Real Quantum Mechanics}

In this section we discuss the mapping of the formulations of the Schrodinger equation in complex Hilbert space to the real K\"ahler space. This procedure will be called realification. Here we breifly remind the definition of the real K\"ahler space which is real Hilbert space equipped with the symplectic structure.  The Real K\"ahler space is a quadruplet case $\mathcal{K}=\{\mathcal{V}=\mathbb{H}\oplus\mathbb{H},g,\mathcal{J},\omega\}$, here $\mathbb{H}$ is a real Hilbert space, $g$ is a inner product on $\mathcal{V}$, $\omega(.,.)=g(J.,.)$ is symplectic form on $\mathcal{V}$.

Consider the Schrödinger equation in a complex Hilbert space $\mathbb{L}^2(X,d\mu)$
\begin{equation}
    i\dot{\psi}=\widehat{H}\psi.
\end{equation}
with $\psi$ belong to the measurable space .
We introduce the operation of realification by mapping of $\psi=U+iV$ , where $U,V$ are real valued functions, into  $ \phi
:=(U,V)^T.$ Under realification, multiplication by \(i\) is represented by the
real matrix
\begin{equation}
    \mathcal{R}(i)=\mathcal{J},
    \qquad
    \mathcal{J}= \begin{pmatrix}
        0 & -\mathbbm{1}\\\mathbbm{1}&0
    \end{pmatrix},
\end{equation}
which satisfies $\mathcal{J}^2=- \mathbbm{1}$. 
If $\mathcal{R}(\widehat H)=\mathcal{H}$ then the realification of the Schrödinger equation takes the form
\begin{equation}
    \mathcal{J}\dot{\phi}=\mathcal{H}\phi.
\end{equation}

\subsection{Finite-Dimensional Case}

The finite-dimensional case is particularly important in quantum
information theory \cite{Ohya}. In this subsection we show that  more general Schr\"odinger equation in finite dimensional Hilbert space is equivalent to the classical Hamiltonian for the harmonic oscillator. Consider the Hilbert space $\mathbb{C}^{N}$, with Hamiltonian $\widehat H=H_{ab}$, where $a,b=1,\ldots,N$ , in the form of Hermitian matrix. Any Hermitian matrix can be written in the form
\begin{equation}
    H_{ab}=K_{ab}+iA_{ab},
\end{equation}
where $K_{ab}$ real symmetric matrix and $A_{ab}$ is a real antisymmetric matrix. Consider the Schrödinger equation where is summation over repeated indices is assumed
\begin{equation}
    i\dot{\psi}_a=H_{ab}\psi_b.
    \end{equation} 
Now, we substitute ${H}_{ab} =K_{ab}+iA_{ab}$ and $\psi_a=U_a+iV_a$, giving us 
\begin{equation}
    i\dot U_a-\dot V_a=K_{ab}U_a+iK_{ab}V_a+iA_{ab}U_a -A_{ab}V_a
\end{equation}
Equating the real and imaginary parts gives
\begin{align}
    \dot U_a &=A_{ab}U_b+K_{ab}V_b,
    \label{eq:U-dynamics}\\
    \dot V_a &=-K_{ab}U_b+A_{ab}V_b.
    \label{eq:V-dynamics}
\end{align}
We may compare these equations with the realified form.
Now let us proof that these Schrödinger equations are equivalent with the realified Schrödinger equation. 
 \begin{equation}
    \mathcal{J}\dot\phi= \mathcal{H} \phi,
\end{equation}
The realification of the Hermitian matrix is
\begin{equation}
    \mathcal{H}
    =
    \mathcal{R}(\widehat H)
    =
    \begin{pmatrix}
        K & -A\\
        A & K
    \end{pmatrix}.
\end{equation}
Since \(\mathcal{J}^{-1}=-\mathcal{J}\), the realified equation becomes
\begin{align}
    \dot{\phi}
    =\begin{pmatrix}
    \dot U\\\dot V
\end{pmatrix} &=
    -\mathcal{J}\mathcal{H}\phi
    \nonumber\\
    &=\begin{pmatrix}
        A & K\\
        -K & A
    \end{pmatrix}
    \begin{pmatrix}
        U\\
        V
    \end{pmatrix}=\begin{pmatrix}
        AU+KV \\-KU+AV
    \end{pmatrix}.
\end{align}
    
Moreover, defining the canonical coordinates \(q_a:=U_a\) and \(p_a:=V_a\), the equations of motion take the canonical Hamiltonian form
\begin{equation}
\dot q_a=\frac{\partial H}{\partial p_a}
=K_{ab}p_b+A_{ab}q_b,
\qquad
\dot p_a=-\frac{\partial H}{\partial q_a}
=-K_{ab}q_b+A_{ab}p_b.
\end{equation}
Integrating Hamilton’s equations yields, up to a dynamically irrelevant additive constant,
\begin{equation}
H(q,p)=\frac12K_{ab}\left(q_aq_b+p_ap_b\right)
+p_aA_{ab}q_b.
\end{equation}
Thus, the evolution of the finite-dimensional quantum system is equivalent to that of a classical system governed by this quadratic Hamiltonian.

\section{Spectral theorem in a real K\"ahler Hilbert space}
\label{sec:KahlerSpectral}

The spectral theorem for unbounded self-adjoint operators on real Hilbert
spaces, together with its compatibility with a fixed complex structure, is discussed
 in \cite{Goodrich1972,Moretti:2016csf}.  Here we formulate the result directly in a real
K\"ahler  space.  In finite dimension it reduces to the spectral
decomposition considered in \cite{Arefeva:2025zbx, Arefeva:2026xlr}, while the
projection-valued formulation below also includes continuous spectrum and
unbounded observables.

Let $(\cK,g,\omega,\mathcal J)$ be a real K\"ahler Hilbert space.  Thus
$\|x\|^2=g(x,x)$, and
\be
\omega(x,y)=g(\mathcal Jx,y).
\ee
All adjoints are taken with respect to $g$. We say that a densely defined real-linear operator $\cL$ is self-adjoint and commutes with $\mathcal J$ if
\be
 \cL^*=\cL,\qquad \mathcal JD(\cL)=D(\cL),\qquad \cL \mathcal Jx=\mathcal J\cL x,
 \quad x\in D(\cL).
 \label{eq:KahlerSelfAdjoint}
\ee
For simplicity, the last two conditions will be written as
$\cL\mathcal J=\mathcal J\cL$. 

A Kähler projection valued measure (PVM) is a real orthogonal PVM $E$ on $\mathbb R$ whose projections commute with the complex structure:
\be
E(\Delta)\mathcal J=\mathcal JE(\Delta).
\ee

{\bf Theorem:}
Let $(\cK,g,\omega,\mathcal J)$ be a real K\"ahler  space, and let $\cL:D(\cL)\subset\cK\longrightarrow\cK$ be a densely defined real-linear self-adjoint operator such that $\cL\mathcal J=\mathcal J\cL$.  Then there exists a unique K\"ahler projection-valued measure $E$ on $\mathbb R$ such that
\be
\cL=\int_{\mathbb R}\lambda\,dE(\lambda).
\ee
More precisely, for every $x\in\cK$, define the finite positive Borel
measure
\be
\mu_x(\Delta)
   =g\bigl(x,E(\Delta)x\bigr),
\qquad
\Delta\in{\cal B}(\mathbb R).
\label{eq:KahlerScalarMeasure}
\ee
Then
\bea
D(\cL)
&=&
\left\{
x\in\cK:
\int_{\mathbb R}\lambda^2\,d\mu_x(\lambda)<\infty
\right\},
\label{eq:KahlerSpectralDomain}\\
\|\cL x\|^2
&=&
\int_{\mathbb R}\lambda^2\,d\mu_x(\lambda),
\qquad x\in D(\cL).
\label{eq:KahlerSpectralNorm}
\eea

{\it Proof.}
Applied to the underlying real Hilbert space $(\cK,g)$, the real spectral
theorem \cite[Theorem B.26(b)]{Moretti:2016csf} gives the unique real
orthogonal PVM $E$ and the stated domain and norm identities.  Since
the complex structure $\mathcal J$ of a K\"ahler  space is bounded and
$\mathcal J\cL\subset\cL\mathcal J$,Theorem B.26(c)(ii) of 
\cite{Moretti:2016csf} gives
\be
 E(\Delta)\mathcal J
 =
 \mathcal JE(\Delta)
\ee
for every $\Delta\in{\cal B}(\mathbb R)$.  Hence $E_{\cL}$ is a K\"ahler
PVM, and its uniqueness follows from the uniqueness of the real spectral
PVM. \hfill$\square$

In particular, every spectral subspace $E_{\cL}(\Delta)\cK$ is invariant under $\mathcal J$. Hence every finite-dimensional eigenspace has even real dimension. This recovers the doubling of spectral multiplicities in the real representation found in the finite-dimensional theory.

For bounded normal operators on real Hilbert spaces, a related $C^*$-algebraic formulation of spectral theory was given in \cite{AgrawalKulkarni1994}.

The spectral theorem also has a multiplication-operator formulation. In the separable case, a self-adjoint operator on a real Hilbert space is orthogonally equivalent to an operator of multiplication by a real measurable function on a suitable real $L^2$-space. In the present Kähler setting, the corresponding spectral subspaces are invariant under $\mathcal J$.

The same PVM also defines the real Borel functional calculus.  For every bounded
Borel function $f:\mathbb R\to\mathbb R$, set
\be
 f(\cL)=\int_{\mathbb R}f(\lambda)\,dE_{\cL}(\lambda).
 \label{eq:KahlerFunctionalCalculus}
\ee
The resulting operator is bounded and self-adjoint.  Since every spectral projection commutes with $\mathcal J$, so does $f(\cL)$.

\section{Stone Theorem in a real K\"ahler Hilbert space}

In this section, we formulate Stone's theorem for strongly continuous one-parameter orthogonal groups that preserve a real Kähler structure. We then apply it to the real Kähler formulation of Schrödinger evolution. We first describe the general construction and then discuss its finite and infinite dimensional realizations.

\begin{theorem}
\label{thm:Kahler-Stone}
Let $(\mathcal K,g,\omega,\mathcal J)$ be a real
K\"ahler  space. A family
$\{S_t\}_{t\in\mathbb R}$ of bounded real-linear operators is a
strongly continuous one-parameter group of orthogonal operators
satisfying
\begin{equation}
    S_t\mathcal J=\mathcal J S_t,
    \qquad t\in\mathbb R,
\end{equation}
if and only if
\begin{equation}
    S_t=e^{t\mathcal G},
    \qquad t\in\mathbb R,
\end{equation}
for a unique skew-adjoint operator $\mathcal G$ satisfying
\begin{equation}
    \mathcal J D(\mathcal G)=D(\mathcal G),
    \qquad
    \mathcal G\mathcal J\phi
    =
    \mathcal J\mathcal G\phi,
    \qquad
    \phi\in D(\mathcal G).
\end{equation}
In this case, every \(S_t\) preserves the Kähler structure: it preserves the inner product \(g\), commutes with the complex structure \(\mathcal J\), and therefore also preserves the symplectic form 
\begin{equation}
    \omega(S_t\phi,S_t\xi)
    =
    \omega(\phi,\xi),
    \qquad
    \phi,\xi\in\mathcal K.
\end{equation}
\end{theorem}

\textit{Proof:}
By the real Stone theorem, see Theorem A.3 from
\cite{FigotinSchenker2007}, there exists a unique
skew-adjoint operator $\mathcal G$ such that
\begin{equation}
 S_t=e^{t\mathcal G}.
\end{equation}

For $\phi\in D(\mathcal G)$, the boundedness of $\mathcal J$ and the
relation $S_t\mathcal J=\mathcal JS_t$ give
\begin{equation}
 \mathcal G\mathcal J\phi
 =
 \lim_{t\to0}\frac{S_t\mathcal J\phi-\mathcal J\phi}{t}
 =
 \mathcal J\lim_{t\to0}\frac{S_t\phi-\phi}{t}
 =
 \mathcal J\mathcal G\phi.
\end{equation}
Thus $\mathcal JD(\mathcal G)\subseteq D(\mathcal G)$.  Since
$\mathcal J^{-1}=-\mathcal J$, equality follows $\mathcal JD(\mathcal G)=D(\mathcal G).$

Conversely, if $\mathcal G$ is skew-adjoint and commutes with $\mathcal J$,
the real Stone theorem gives $S_t=e^{t\mathcal G}$; since
$\mathcal J^{-1}S_t\mathcal J$ has the same generator $\mathcal G$,
uniqueness implies $S_t\mathcal J=\mathcal JS_t$.
Finally, orthogonality and commutation with $\mathcal J$ imply preservation of symplectic form
\begin{equation}
 \omega(S_t\phi,S_t\xi)
 =
 g(\mathcal JS_t\phi,S_t\xi)
 =
 g(S_t\mathcal J\phi,S_t\xi)
 =g(\mathcal J\phi,\xi)
 =
 \omega(\phi,\xi).
\end{equation}

\hfill$\square$

\subsection{General form of evolution equation}

Let $H_{\mathbb C}$ be a possibly unbounded self-adjoint
Hamiltonian, and let $\mathcal H$ denote its realification. Since
realification preserves adjoints, $\mathcal H$ is self-adjoint. Moreover, since \(D(H_{\mathbb C})\) is a complex vector subspace,
\(iD(H_{\mathbb C})=D(H_{\mathbb C}),\)
which under realification becomes
\(\mathcal JD(\mathcal H)=D(\mathcal H).\) Complex linearity of
$H_{\mathbb C}$,
$H_{\mathbb C}\bigl(i\psi\bigr)
=
iH_{\mathbb C}\psi$, gives
\begin{equation}
\mathcal H\mathcal J\phi
=
\mathcal J\mathcal H\phi, \quad 
\qquad
\phi\in D(\mathcal H).
\end{equation}
Thus $\mathcal H$ commutes with $\mathcal J$

Consider the realified Schr\"odinger equation
\be
 \dot{\phi}(t)=-\mathcal J\mathcal H\phi(t),
 \qquad \phi(0)=\phi_0.
 \label{eq:real-Schrodinger-equation}
\ee
Its generator is $\mathcal G=-\mathcal J\mathcal H$ on
$D(\mathcal G)=D(\mathcal H)$.  Since
$\mathcal H$ is self-adjoint, $\mathcal J$ is skew-ad joint, and$\mathcal H$ commutes with $\mathcal J$,
\be
 \mathcal G^*
 =
 \mathcal H\mathcal J
 =
 \mathcal J\mathcal H
 =
 -\mathcal G,
 \qquad
 \mathcal G\mathcal J=\mathcal J\mathcal G.
\ee
Hence Theorem~\ref{thm:Kahler-Stone} yields the strongly continuous
one-parameter group of orthogonal operators
\be
 S_t=e^{t\mathcal G}=e^{-t\mathcal J\mathcal H},
\ee
which commutes with $\mathcal J$, and preserves
$\omega$.  For every \(\phi_0\in\mathcal K_{\mathbb R}\), the family
\begin{equation}
\phi(t)=S_t\phi_0
\end{equation}defines a strongly continuous evolution. If
\(\phi_0\in D(\mathcal H)\), then \(\phi(t)\) is differentiable and is the unique classical solution of \eqref{eq:real-Schrodinger-equation}.

\begin{remark}
As a simple infinite-dimensional example, let $\mathscr H_{\mathbb C}=L^2(\mathbb R^3;\mathbb C^2),$
\begin{equation}
h=-\Delta+V,\qquad D(h)=H^2(\mathbb R^3),
\end{equation}
with \(V\in L^\infty(\mathbb R^3;\mathbb R)\).
For a bounded external field
\begin{equation}
\mathbf B=(b_1,b_2,b_3),
\qquad
b_i\in L^\infty(\mathbb R^3;\mathbb R),
\end{equation}
consider the Hamiltonian
\begin{equation}
H_{\mathbb C}
=
hI_2-\mu b_i\sigma_i
=
\begin{pmatrix}
h-\mu b_3 & -\mu b_1+i\mu b_2\\
-\mu b_1-i\mu b_2 & h+\mu b_3
\end{pmatrix},
\qquad
\mu\in\mathbb R.
\label{eq:Pauli-simple-Hamiltonian}
\end{equation}
Since \(-\mu b_i\sigma_i\) is bounded and self-adjoint,
\(H_{\mathbb C}\) is self-adjoint on $D(H_{\mathbb C})=H^2(\mathbb R^3;\mathbb C^2).$
Writing
\begin{equation}
H_{\mathbb C}=P+iQ,
\qquad
P=
\begin{pmatrix}
h-\mu b_3&-\mu b_1\\
-\mu b_1&h+\mu b_3
\end{pmatrix},
\qquad
Q=
\mu b_2
\begin{pmatrix}
0&1\\
-1&0
\end{pmatrix},
\end{equation}
we have \(P^*=P\) and \(Q^*=-Q\) on
\(L^2(\mathbb R^3;\mathbb R^2)\). Moreover, \(Q\) is bounded and $D(P)=H^2(\mathbb R^3;\mathbb R^2).$ By the realification formula established above,
\begin{equation}
\mathcal H=
\begin{pmatrix}
P&-Q\\
Q&P
\end{pmatrix},
\qquad
\mathcal G=-\mathcal J\mathcal H
=
\begin{pmatrix}
Q&P\\
-P&Q
\end{pmatrix},
\qquad
D(\mathcal G)=D(P)\oplus D(P).
\end{equation}
Since \(Q\) is bounded,
\(D(\mathcal G^*)=D(\mathcal G)\), and
\begin{equation}
\mathcal G^*
=
\begin{pmatrix}
-Q&-P\\
P&-Q
\end{pmatrix}
=
-\mathcal G.
\end{equation}
Thus \(\mathcal G\) is skew-adjoint. By construction, $\mathcal G\mathcal J=\mathcal J\mathcal G$.
Consequently, Theorem~\ref{thm:Kahler-Stone} yields the strongly continuous orthogonal group
\begin{equation}
S_t=e^{t\mathcal G}=e^{-t\mathcal J\mathcal H},
\end{equation}
which commutes with \(\mathcal J\) and therefore preserves the full K\"ahler structure.
\end{remark}

\section{Extension into a Classically Forbidden Region}

In this section, we examine the evanescent penetration of a stationary scattering solution into a classically forbidden region in the real Kähler formulation. For a semi-infinite potential step,  we are showing that the phenomena of squeezing into a classically forbidden region also possible in the formulation using real number.

Consider a quantum particle moving on the real line in a real valued
potential $V(x)$. In the standard complex formulation, its Hamiltonian
acts on $L^2(\mathbb R;\mathbb C)$ and, in units where $\hbar=1$, is
\begin{equation}
    \widehat H
    =
    -\frac{1}{2m}\frac{d^2}{dx^2}
    +
    V(x).
    \label{eq:step-complex-Hamiltonian}
\end{equation}
Since $\widehat H$ has real coefficients, its realification on
$L^2(\mathbb R;\mathbb R^2)$ is
 $\mathcal H
    =
    \begin{pmatrix}
        \widehat H&0\\
        0&\widehat H
    \end{pmatrix}.$

For a state of definite energy $E$, we use the stationary ansatz
\begin{equation}
    \phi(t,x)
    =
    e^{-Et\mathcal J}\phi_E(x),
    \qquad
    \phi_E(x)
    =
    \begin{pmatrix}
        u_E(x)\\
        v_E(x)
    \end{pmatrix}.
    \label{eq:real-stationary-ansatz}
\end{equation}
Substituting this ansatz into
$\partial_t\phi=-\mathcal J\mathcal H\phi$ and using
$[\mathcal H,\mathcal J]=0$ gives
\begin{equation}
    -E\mathcal J e^{-Et\mathcal J}\phi_E
    =
    -\mathcal J e^{-Et\mathcal J}\mathcal H\phi_E,
\end{equation}
shows that $\phi_E$ satisfies the eigenvalue equation 
\begin{equation}
    \mathcal H\phi_E=E\phi_E.
\end{equation}
Thus the realified stationary Schr\"odinger equation is
\begin{equation}
    -\frac{1}{2m}\phi_E''(x)
    +
    V(x)\phi_E(x)
    =
    E\phi_E(x).
    \label{eq:real-stationary-step-equation}
\end{equation}

  Because $\mathcal H=\widehat H\oplus\widehat H$, the two real
components $u_E$ and $v_E$ independently satisfy the same scalar
stationary Schr\"odinger equation. Consequently, every scalar
eigenfunction $f_E$ of $\widehat H$ gives two linearly independent real
eigenvectors,
\begin{equation}
    \begin{pmatrix}f_E\\0\end{pmatrix},
    \qquad
    \begin{pmatrix}0\\f_E\end{pmatrix},
\end{equation}
with the same eigenvalue $E$. Thus the multiplicity of each eigenvalue
is doubled in the real representation, as discussed in
\cite{Arefeva:2026xlr}. Since the component equations have the same form as the
usual stationary Schr\"odinger equation, the standard potential-step
and tunnelling calculations can be performed componentwise using real
variables.

The local inner product defines the probability density in the real formulation
\begin{equation}
    \rho(t,x)
    =
    g(\phi(t,x),\phi(t,x))_x
    =
    U(t,x)^2+V(t,x)^2.
    \label{eq:real-probability-density}
\end{equation}
Under the identification \(\psi=u+iv\), this expression is precisely \(\rho=|\psi|^2\).
For a state of definite energy, the operator $e^{-Et\mathcal J}$ acts at each point $x$ as an orthogonal transformation. Therefore,
\begin{equation}
    \rho(t,x)
    =
    g(\phi(t,x),\phi(t,x))_x
    =
    g(\phi_E(x),\phi_E(x))_x.
\end{equation}
Thus the probability density of the stationary state is independent
of time.For a normalized state, integrating \(\rho\) over a region gives the probability of finding the particle in that region. The stationary scattering solutions considered below are generalized
eigenfunctions and are not normalizable on the full real line.
Consequently, their reflection and transmission properties are
characterized by ratios of probability currents.

In the real K\"ahler formulation, the probability current is obtained
from the local symplectic form:
\begin{equation}
    j(t,x)
    =
    \frac{1}{m}
    \omega_x\left(
        \phi(t,x),
        \frac{\partial\phi(t,x)}{\partial x}
    \right),
    \label{eq:real-probability-current}
\end{equation}
Since $\omega_x(\phi,\xi)
    =
    g(\mathcal J\phi,\xi)_x,$ the current takes the explicit form
\begin{equation}
    j(t,x)
    =
    \frac{1}{m}
    \left(
        U(t,x)\frac{\partial V(t,x)}{\partial x}
        -
        V(t,x)\frac{\partial U(t,x)}{\partial x}
    \right).
\end{equation}
This is exactly the usual complex probability current. Moreover,  probability conservation is expressed by the continuity equation
\begin{equation}
\frac{\partial}{\partial t}g_x(\phi,\phi)+\frac{\partial}{\partial x}\left(\frac{1}{m}\omega_x(\phi,\frac{\partial\phi(t,x)}{\partial x})\right)=0.
\end{equation}

For the stationary state~\eqref{eq:real-stationary-ansatz},
\begin{equation}
    \frac{\partial\phi(t,x)}{\partial x}
    =
    e^{-Et\mathcal J}\phi_E'(x).
\end{equation}
Since $e^{-Et\mathcal J}$ preserves the local symplectic form, it follows that
\begin{equation}
\begin{aligned}
    j(t,x)
    &=
    \frac{1}{m}
    \omega_x\left(
        e^{-Et\mathcal J}\phi_E(x),
        e^{-Et\mathcal J}\phi_E'(x)
    \right)\\
    &=
    \frac{1}{m}
    \omega_x\left(
        \phi_E(x),
        \phi_E'(x)
    \right).
\end{aligned}
\end{equation}
Thus the probability current of a stationary state with definite energy is also independent of time.

We now apply this formulation to a potential step and examine the
extension of a stationary state into a classically forbidden region.
Let
\begin{equation}
    V(x)
    =
    \begin{cases}
        0,   & x<0,\\
        V_0, & x>0,
    \end{cases}
    \qquad
    0<E<V_0.
    \label{eq:real-potential-step}
\end{equation}
The region \(x<0\) is classically allowed, whereas \(x>0\) is
classically forbidden. For \(x<0\), the Schr\"dinger equation becomes
\begin{equation}
    \phi_E''+k^2\phi_E=0,
    \qquad
    k=\sqrt{2mE}.
    \label{eq:real-allowed-equation}
\end{equation}
Since $e^{kx\mathcal J}
    =
    I\cos(kx)+\mathcal J\sin(kx),$ the solution can be written as
\begin{equation}
    \phi_{\mathrm I}(x)
    =
    e^{kx\mathcal J}\mathbf a
    +
    e^{-kx\mathcal J}\mathbf b,
    \qquad
    x<0.
    \label{eq:real-allowed-solution}
\end{equation}
Here \(\mathbf a,\mathbf b\in\mathbb R^2\) are the amplitude vectors. The factors \(e^{kx\mathcal J}\) and
\(e^{-kx\mathcal J}\) are the real representations of \(e^{ikx}\)
and \(e^{-ikx}\), respectively.

For \(x>0\), equation~\eqref{eq:real-stationary-step-equation}
becomes
\begin{equation}
    \phi_E''-\kappa^2\phi_E=0,
    \qquad
    \kappa
    =\sqrt{2m(V_0-E)}.
    \label{eq:real-forbidden-equation}
\end{equation}
Its general solution is
\begin{equation}
    \phi_{\mathrm {II}}(x)
    =
    e^{\kappa x}\mathbf c_+
    +
    e^{-\kappa x}\mathbf c_-.
    \label{eq:real-forbidden-general-solution}
\end{equation}
Boundedness as \(x\to+\infty\) requires
\(\mathbf c_+=0\). Hence
\begin{equation}
    \phi_{\mathrm {II}}(x)
    =
    e^{-\kappa x}\mathbf c,
    \qquad
    x>0,
    \label{eq:real-evanescent-solution}
\end{equation}
where \(\mathbf c\in\mathbb R^2\) is the forbidden region amplitude.

Because the potential has only a finite jump at \(x=0\), contains no delta-function contribution, and the particle mass is constant, integrating the stationary Schrödinger equation across the boundary shows that both \(\phi_E\) and \(\phi_E'\) are continuous at \(x=0\).
\begin{equation}
    \phi_{\mathrm I}(0)
    =
    \phi_{\mathrm {II}}(0),
    \qquad
    \phi_{\mathrm I}'(0)
    =
    \phi_{\mathrm {II}}'(0).
    \label{eq:real-step-boundary-conditions}
\end{equation}
Matching the solutions at the boundary gives the following relations between the amplitudes:
\begin{equation}
    \mathbf a+\mathbf b=\mathbf c,
     \qquad 
    k\mathcal J(\mathbf a-\mathbf b)
    =
    -\kappa\mathbf c.
\end{equation}
Eliminating \(\mathbf b\)gives the forbidden region amplitude 
\begin{equation}
    \mathbf c
    =
    \mathsf P\mathbf a,
    \qquad
    \mathsf P
    =
    \frac{2k}{k^2+\kappa^2}
    \left(
        kI-\kappa\mathcal J
    \right).
    \label{eq:real-penetration-matrix}
\end{equation}
Using \(\mathbf b=\mathbf c-\mathbf a\), the reflected amplitude is
\begin{equation}
    \mathbf b
    =
    \mathsf R\mathbf a,
    \qquad
    \mathsf R
    =
    \frac{
        (k^2-\kappa^2)I
        -
        2k\kappa\mathcal J
    }{
        k^2+\kappa^2
    }.
    \label{eq:real-reflection-matrix}
\end{equation}
The matrix \(\mathsf R\) is orthogonal. Therefore,
\begin{equation}
    \|\mathbf b\|
    =
    \|\mathbf a\|.
    \label{eq:amplitude-norm-equality}
\end{equation}

In the allowed region, the derivative of the stationary state is
\begin{equation}
    \phi_{\mathrm I}'(x)
    =
    k\mathcal J e^{kx\mathcal J}\mathbf a
    -
    k\mathcal J e^{-kx\mathcal J}\mathbf b.
\end{equation}
Substitution into the probability current gives
\begin{equation}
\begin{aligned}
    j_{\mathrm I}
    &=
    \frac{1}{m}
    \omega_x\left(
        \phi_{\mathrm I},
        \phi_{\mathrm I}'
    \right)\\
    &=
    \frac{k}{m}
    \left(
        \|\mathbf a\|^2-\|\mathbf b\|^2
    \right)=0.
\end{aligned}
\end{equation}
Because $|\mathbf b|=|\mathbf a|$, the net probability current in the allowed region vanishes.
The probability density in this region is
\begin{equation}
\begin{aligned}
    \rho_{\mathrm I}(x)
    &=
   g(
        \phi_{\mathrm I}(x),
        \phi_{\mathrm I}(x)
    )_x\\
    &=
    \|\mathbf a\|^2+\|\mathbf b\|^2
    +
    2g(
        e^{kx\mathcal J}\mathbf a,
        e^{-kx\mathcal J}\mathbf b
    )_x.
\end{aligned}
\end{equation}
The last term describes interference between the incident and reflected wave components.

In the forbidden region, the probability density is
\begin{equation}
\rho_{\mathrm {II}}(x)
    =
    g(
        \phi_{\mathrm {II}}(x),
        \phi_{\mathrm {II}}(x)
    )_x=
    e^{-2\kappa x}\|\mathbf c\|^2.
\end{equation}
Therefore, the probability density is nonzero immediately beyond the
step but decreases exponentially as $x$ increases.

Since
\begin{equation}
    \phi_{\mathrm {II}}'(x)
    =
    -\kappa\phi_{\mathrm {II}}(x),
\end{equation}
the probability current in the forbidden region is
\begin{equation}
\begin{aligned}
    j_{\mathrm {II}}
    &=
    \frac{1}{m}
    \omega_x\left(
        \phi_{\mathrm {II}},
        \phi_{\mathrm {II}}'
    \right)\\
    &=
    -\frac{\kappa}{m}
    \omega_x\left(
        \phi_{\mathrm {II}},
        \phi_{\mathrm {II}}
    \right)\\
    &=
    0.
\end{aligned}
\end{equation}
Thus the state extends into the classically forbidden region, although
it carries no probability current there. The potential step therefore
produces complete reflection together with an exponentially decreasing
probability density beyond the boundary.

\section{Tunneling in the Real Kähler Formulation}

In this section, we consider tunneling through a finite rectangular potential barrier in the real Kähler formulation. Although the state is not expressed using complex valued coordinates, the resulting transmission coefficient agrees with that obtained in the standard complex formulation.

The finite barrier is defined by
\begin{equation}
    V(x)
    =
    \begin{cases}
        0,   & x<0,\\
        V_0, & 0\leq x\leq L,\\
        0,   & x>L,
    \end{cases}
    \qquad
    0<E<V_0.
    \label{eq:finite-barrier}
\end{equation}
The regions $x<0$ and $x>L$ are classically allowed, while
$0<x<L$ is classically forbidden. The  stationary spatial solution in the
first region is
\begin{equation}
    \phi_{\mathrm I}(x)
    =
    e^{kx\mathcal J}\mathbf a
    +
    e^{-kx\mathcal J}\mathbf b,
    \qquad
    k=\sqrt{2mE}.
    \label{eq:finite-barrier-region-I}
\end{equation}
Here $\mathbf a$ and $\mathbf b$ are the real amplitude vectors of the right-moving incident wave and the left-moving reflected wave, respectively.

Inside the barrier, the stationary equation is
\begin{equation}
    \phi_{\mathrm {II}}''-\kappa^2\phi_{\mathrm {II}}=0,
    \qquad
    \kappa=\sqrt{2m(V_0-E)}.
\end{equation}
Since this region is finite, both exponential solutions must be
retained:
\begin{equation}
    \phi_{\mathrm {II}}(x)
    =
    e^{\kappa x}\mathbf c
    +
    e^{-\kappa x}\mathbf d,
    \qquad
    0<x<L.
    \label{eq:finite-barrier-region-II}
\end{equation}
Here $\mathbf c$ and $\mathbf d$ are the real amplitude vectors of the two linearly independent exponential solutions.

In the third region, we impose the scattering condition that no wave is incident from the right. Consequently, the solution contains only a right-moving transmitted wave:
\[
\phi_{\mathrm{III}}(x)=
e^{k(x-L)\mathcal J}\mathbf f,
\qquad x>L,
\label{eq:finite-barrier-region-III}
\]
where \(\mathbf f\) is the transmitted amplitude vector.

Continuity of the state and its derivative at $x=0$ gives
\begin{equation}
    \mathbf a+\mathbf b
    =
    \mathbf c+\mathbf d,
    \qquad
    k\mathcal J(\mathbf a-\mathbf b)
    =
    \kappa(\mathbf c-\mathbf d).
    \label{eq:finite-barrier-boundary-zero}
\end{equation}
At $x=L$, continuity of the stationary spatial state and its derivative gives
\begin{equation}
    e^{\kappa L}\mathbf c
    +
    e^{-\kappa L}\mathbf d
    =
    \mathbf f,
\end{equation}
\begin{equation}
    \kappa
    \left(
        e^{\kappa L}\mathbf c
        -
        e^{-\kappa L}\mathbf d
    \right)
    =
    k\mathcal J\mathbf f.
    \label{eq:finite-barrier-boundary-L}
\end{equation}
Solving the conditions at $x=L$ for $\mathbf c$ and $\mathbf d$ gives
\begin{equation}
    \mathbf c
    =
    \frac{e^{-\kappa L}}{2}
    \left(
        I+\frac{k}{\kappa}\mathcal J
    \right)\mathbf f,
\end{equation}
\begin{equation}
    \mathbf d
    =
    \frac{e^{\kappa L}}{2}
    \left(
        I-\frac{k}{\kappa}\mathcal J
    \right)\mathbf f.
\end{equation}
Substituting these expressions for $\mathbf c$ and $\mathbf d$ into the matching conditions at $x = 0$ and eliminating $\mathbf b$ gives
\begin{equation}
    \mathbf a
    =
    \left[
        \cosh(\kappa L)I
        +
        \frac{\kappa^2-k^2}{2k\kappa}
        \sinh(\kappa L)\mathcal J
    \right]\mathbf f.
    \label{eq:incident-transmitted-relation}
\end{equation}
Since $\mathcal J^2=-I$, a matrix of the form
$\alpha I+\beta\mathcal J$ has the inverse
\begin{equation}
    \left(\alpha I+\beta\mathcal J\right)^{-1}
    =
    \frac{\alpha I-\beta\mathcal J}{\alpha^2+\beta^2}.
\end{equation}
Therefore,
\begin{equation}
    \mathbf f=\mathsf T_L\mathbf a,
\end{equation}
where $\mathsf T_L$ is the transmission amplitude operator given by
\begin{equation}
    \mathsf T_L
    =
    \frac{
        \cosh(\kappa L)I
        -
        q\sinh(\kappa L)\mathcal J
    }{
        \cosh^2(\kappa L)
        +
        q^2\sinh^2(\kappa L)
    },
    \qquad
    q=\frac{\kappa^2-k^2}{2k\kappa}.
    \label{eq:finite-transmission-matrix}
\end{equation}

We now use this amplitude relation to determine the physically
measurable transmission coefficient. Because the potential vanishes
in both asymptotic regions, the incident and transmitted waves have
the same wave number. Their
probability currents are therefore
\begin{equation}
    j_{\mathrm{in}}
    =
    \frac{k}{m}\|\mathbf a\|^2,
    \qquad
    j_{\mathrm{tr}}
    =
    \frac{k}{m}\|\mathbf f\|^2.
\end{equation}
Consequently,
\begin{equation}
    T
    =
    \frac{j_{\mathrm{tr}}}{j_{\mathrm{in}}}
    =
    \frac{\|\mathbf f\|^2}{\|\mathbf a\|^2}.
    \label{eq:transmission-current-ratio}
\end{equation}

The amplitude relation $\mathbf f=\mathsf T_L\mathbf a$ gives
\begin{equation}
    \|\mathbf f\|^2
    =
    g\left(
        \mathbf a,
        \mathsf T_L^*\mathsf T_L\mathbf a
    \right).
\end{equation}
Thus, the transmitted to incident norm ratio is determined by the
positive operator $\mathsf T_L^*\mathsf T_L$. Direct multiplication
gives
\begin{equation}
    \mathsf T_L^*\mathsf T_L
    =
    \frac{1}{
        \cosh^2(\kappa L)
        +
        q^2\sinh^2(\kappa L)
    }I.
\end{equation}
Since this operator is proportional to the identity, the transmission
coefficient is independent of the direction of the incident amplitude
vector. Substituting the preceding identity into
Eq.~\eqref{eq:transmission-current-ratio} yields
\begin{equation}
\begin{aligned}
    T
    &=
    \frac{1}{
        \cosh^2(\kappa L)
        +
        q^2\sinh^2(\kappa L)
    }=
    \left[
        1+
        \frac{(k^2+\kappa^2)^2}{4k^2\kappa^2}
        \sinh^2(\kappa L)
    \right]^{-1}.
\end{aligned}
\label{eq:finite-barrier-transmission}
\end{equation}
Since
\begin{equation}
    k^2=2mE,
    \qquad
    \kappa^2=2m(V_0-E),
\end{equation}
this may also be written as
\begin{equation}
    T
    =
    \left[
        1+
        \frac{V_0^2}{4E(V_0-E)}
        \sinh^2(\kappa L)
    \right]^{-1}.
    \label{eq:finite-barrier-transmission-energy}
\end{equation}

The reflected wave propagates in the negative $x$-direction and
therefore carries the current
\begin{equation}
    j_{\mathrm{ref}}
    =
    -\frac{k}{m}\|\mathbf b\|^2.
\end{equation}
Its reflection coefficient is consequently
\begin{equation}
    R
    =
    -\frac{j_{\mathrm{ref}}}{j_{\mathrm{in}}}
    =
    \frac{\|\mathbf b\|^2}{\|\mathbf a\|^2}.
\end{equation}

For a stationary scattering state, the probability current is
independent of position. The total current in region I must therefore
equal the transmitted current in region III:
\begin{equation}
    j_{\mathrm{in}}+j_{\mathrm{ref}}
    =
    j_{\mathrm{tr}}.
\end{equation}
Dividing this relation by $j_{\mathrm{in}}$ gives
\begin{equation}
    R+T=1,
\end{equation}
as required by conservation of probability current.

We next examine how the same conserved current is represented inside
the classically forbidden region. Substituting the two evanescent
components of $\phi_{\mathrm{II}}$ into the current gives
\begin{equation}
\begin{aligned}
    j_{\mathrm{II}}
    &=
    \frac{1}{m}
    \omega_x\left(
        \phi_{\mathrm{II}},
        \phi_{\mathrm{II}}'
    \right)\\
    &=
    -\frac{2\kappa}{m}
    \omega_x(\mathbf c,\mathbf d).
\end{aligned}
\end{equation}
Using the expressions for $\mathbf c$ and $\mathbf d$ obtained from
the matching conditions, one finds
\begin{equation}
    \omega_x(\mathbf c,\mathbf d)
    =
    -\frac{k}{2\kappa}\|\mathbf f\|^2.
\end{equation}
It follows that
\begin{equation}
    j_{\mathrm{II}}
    =
    \frac{k}{m}\|\mathbf f\|^2
    =
    j_{\mathrm{tr}}.
\end{equation}

Neither exponential term contributes to the probability current separately. The nonzero current arises entirely from the mixed term involving both exponential components. This is the real Kähler counterpart of the relative phase
between the corresponding complex amplitudes.

Finally, in the wide-barrier regime $\kappa L\gg 1$, one has
$\sinh^2(\kappa L)\simeq e^{2\kappa L}/4$. The transmission
coefficient therefore takes the asymptotic form
\begin{equation}
    T
    \simeq
    \frac{16E(V_0-E)}{V_0^2}
    e^{-2\kappa L}.
\end{equation}
Thus, transmission through a sufficiently wide barrier is
exponentially suppressed but remains nonzero. The real Kähler
formulation therefore reproduces both probability-current
conservation and the characteristic exponential behavior of quantum
tunneling.

\section{Scattering Theory in Real K\"ahler Space}

In this section, we develop scattering theory for evolution equation on a real Kähler space. We first establish the existence of the real wave operators by a real version of Cook's method \cite{Reed:1979ne}. Under the small-Rollnik hypothesis, we then discuss asymptotic completeness and construct the scattering operator, showing that its real Kähler representation reproduces the phase shifts and differential cross sections of the usual complex formulation.

Let $\mathcal{K}=L^2(\mathbb R^3;\mathbb R)\oplus L^2(\mathbb R^3;\mathbb R)$  equipped with standard Kähler structure be the real K\"ahler Hilbert space. Let self adjoint operator $h_0$ have $D(h_0)=H^2(\mathbb R^3;\mathbb R)$ and let \(V\in L^\infty(\mathbb R^3;\mathbb R)\) be a measurable real-valued function. For the Cook method below, assume that
\begin{equation}
D(h)=D(h_0),
\qquad
h=h_0+V,
\end{equation}
and  \(h\) is self-adjoint. 

On the real Kähler space \(\mathcal K\), define the free and interacting Hamiltonians by
\begin{equation}
\mathcal H_0
=
\begin{pmatrix}
h_0&0\\
0&h_0
\end{pmatrix},
\qquad
\mathcal H=\mathcal H_0+\mathcal V
=
\begin{pmatrix}
h&0\\
0&h
\end{pmatrix},
\qquad
\mathcal V
=
\begin{pmatrix}
V&0\\
0&V
\end{pmatrix}.
\label{eq:real-scattering-Hamiltonians}
\end{equation}
Their domains are $D(\mathcal H_0)
=
D(h_0)\oplus D(h_0)$ and $D(\mathcal H)
=
D(h)\oplus D(h)$. 
Consequently, $D(\mathcal H_0)=D(\mathcal H)$.
Because \(h_0\) and \(h\) are self-adjoint, the block operators
\(\mathcal H_0\) and \(\mathcal H\) are self-adjoint on the specified domains.

The repeated diagonal blocks also imply the commutation relations
\begin{equation}
\mathcal H_0\mathcal J
=
\mathcal J\mathcal H_0
,
\qquad
\mathcal H\mathcal J
=
\mathcal J\mathcal H
\end{equation}
Thus, by the real Stone theorem established in Section~4,
\(-\mathcal J\mathcal H_0\) and \(-\mathcal J\mathcal H\) generate the strongly continuous one-parameter groups of orthogonal operators
\begin{equation}
\mathbb U_0(t)
=
e^{-t\mathcal J\mathcal H_0},
\qquad
\mathbb U(t)
=
e^{-t\mathcal J\mathcal H}.
\label{eq:real-scattering-evolution}
\end{equation}
Both groups commute with \(\mathcal J\) and therefore preserve the complete Kähler structure \((g,\omega,\mathcal J)\).

\medskip
\noindent
\textbf{ Real Cook's Method:}
Let $\mathcal D
\subset
D(\mathcal H_0)=D(\mathcal H)$ be dense in \(\mathcal K\). Suppose that, for every
\(\phi\in\mathcal D\),
\begin{equation}
\int_0^\infty
\left(
\left\|
\mathcal V\mathbb U_0(t)\phi
\right\|_{\mathcal K}
+
\left\|
\mathcal V\mathbb U_0(-t)\phi
\right\|_{\mathcal K}
\right)\,dt
<\infty.
\label{eq:real-Cook-condition}
\end{equation}
The first term controls the limit as
\(t\to+\infty\), whereas the second controls the limit as
\(t\to-\infty\). Then the real wave operators
\begin{equation}
\Omega_\pm^{\mathbb R}
:=
\operatorname*{s-lim}_{t\to\pm\infty}
\mathbb U(-t)\mathbb U_0(t)
=
\operatorname*{s-lim}_{t\to\pm\infty}
e^{t\mathcal J\mathcal H}
e^{-t\mathcal J\mathcal H_0}
\label{eq:real-wave-operators}
\end{equation}
exist on \(\mathcal K\).

\textit{Proof.}
For \(\phi\in\mathcal D\), set
\[
F_\phi(t)
:=
\mathbb U(-t)\mathbb U_0(t)\phi.
\]
The free evolution preserves \(D(\mathcal H_0)\), and
\(D(\mathcal H_0)=D(\mathcal H)\). Hence \(F_\phi\) is strongly differentiable. Using
\[
\frac{d}{dt}\mathbb U(-t)
=
\mathbb U(-t)\mathcal J\mathcal H,
\qquad
\frac{d}{dt}\mathbb U_0(t)
=
-\mathcal J\mathcal H_0\mathbb U_0(t),
\]
together with
\(\mathcal H-\mathcal H_0=\mathcal V\), we obtain
\begin{equation}
F_\phi'(t)
=
\mathbb U(-t)\mathcal J\mathcal V
\mathbb U_0(t)\phi.
\label{eq:real-Cook-derivative}
\end{equation}
Since \(\mathbb U(-t)\) and \(\mathcal J\) are orthogonal,
\begin{equation}
\|F_\phi'(t)\|_{\mathcal K}
=
\|\mathcal V\mathbb U_0(t)\phi\|_{\mathcal K}.
\label{eq:real-Cook-derivative-norm}
\end{equation}

For \(t>s>0\), the fundamental theorem of calculus gives
\begin{equation}
\|F_\phi(t)-F_\phi(s)\|_{\mathcal K}
\leq
\int_s^t
\|\mathcal V\mathbb U_0(u)\phi\|_{\mathcal K}\,du.
\label{eq:real-Cook-positive-estimate}
\end{equation}
The first term in \eqref{eq:real-Cook-condition} therefore shows that
\(F_\phi(t)\) is Cauchy as \(t\to+\infty\). Since the negative limit computation follow the same, both limits exist for every
\(\phi\in\mathcal D\). Uniform boundedness extends convergence from the dense set \(\mathcal D\) to all of \(\mathcal K\)

For every \(\psi\in\mathcal K\), $\|\Omega_\pm^{\mathbb R}\psi\|
=
\|\psi\|$. Thus \(\Omega_\pm^{\mathbb R}\) are isometries and
\[
(\Omega_\pm^{\mathbb R})^*\Omega_\pm^{\mathbb R}=I.
\]

Since both evolution groups commute with \(\mathcal J\), Passing to the strong limit gives
\[
\Omega_\pm^{\mathbb R}\mathcal J
=
\mathcal J\Omega_\pm^{\mathbb R}.
\]

For \(\tau\in\mathbb R\), the group property and the change of variable
\(s=t-\tau\) give
\begin{align*}
\mathbb U(\tau)\Omega_\pm^{\mathbb R}
&=
\lim_{t\to\pm\infty}
\mathbb U(\tau-t)\mathbb U_0(t)\\
&=
\lim_{s\to\pm\infty}
\mathbb U(-s)\mathbb U_0(s)\mathbb U_0(\tau)
=
\Omega_\pm^{\mathbb R}\mathbb U_0(\tau).
\end{align*}

Finally, let $\phi\in D(\mathcal H_0)$. From
\begin{equation}
\mathbb U(\tau)\Omega_\pm^{\mathbb R}\phi
=
\Omega_\pm^{\mathbb R}\mathbb U_0(\tau)\phi,
\end{equation}
we obtain
\begin{equation}
\frac{
\mathbb U(\tau)\Omega_\pm^{\mathbb R}\phi
-
\Omega_\pm^{\mathbb R}\phi
}{\tau}
=
\Omega_\pm^{\mathbb R}
\frac{
\mathbb U_0(\tau)\phi-\phi
}{\tau}.
\end{equation}
Since $\phi\in D(\mathcal H_0)$,
\begin{equation}
\lim_{\tau\to0}
\frac{
\mathbb U_0(\tau)\phi-\phi
}{\tau}
=
-\mathcal J\mathcal H_0\phi.
\end{equation}
Therefore, the left-hand side also has a limit. By the definition of the generator of $\mathbb U(\tau)$, this proves that
\begin{equation}
\Omega_\pm^{\mathbb R}\phi\in D(\mathcal H)
\end{equation}
and
\begin{equation}
(-\mathcal J\mathcal H)
\Omega_\pm^{\mathbb R}\phi
=
\Omega_\pm^{\mathbb R}
(-\mathcal J\mathcal H_0)\phi.
\end{equation}
Multiplying both sides by $\mathcal J$ and using
$\mathcal J^2=-I$ and
$\mathcal J\Omega_\pm^{\mathbb R}
=
\Omega_\pm^{\mathbb R}\mathcal J$, we obtain
\begin{equation}
\mathcal H\Omega_\pm^{\mathbb R}\phi
=
\Omega_\pm^{\mathbb R}\mathcal H_0\phi,
\qquad
\phi\in D(\mathcal H_0).
\label{eq:real-wave-Hamiltonian-intertwining}
\end{equation}

Cook's criterion gives the existence of the wave operators under a
time-integrability assumption. To obtain asymptotic completeness, we now
impose an independent smallness condition on the potential.
\paragraph{Small-Rollnik potentials:}A measurable real-valued function \(V\) belongs to the Rollnik class if \cite{Volovich:2020integrability,Sakbaev:2023vfp}
\begin{equation}
\|V\|_{\mathrm R}
:=
\left(
\int_{\mathbb R^3}\int_{\mathbb R^3}
\frac{|V(x)|\,|V(y)|}{|x-y|^2}
\,dx\,dy
\right)^{1/2}
<\infty.
\label{eq:Rollnik-norm}
\end{equation}
For the normalization \(h_0=-\Delta/(2m)\), we assume the
small-Rollnik condition
\begin{equation}
\|V\|_{\mathrm R}<\frac{2\pi}{m}.
\label{eq:small-Rollnik-condition}
\end{equation}

Let us now replace the boundedness hypothesis by the Rollnik hypothesis and redefine \(h\) as the self-adjoint operator defined as the form sum of
\(h_0\) and \(V\), and set
\[
\mathcal H_0=\operatorname{diag}(h_0,h_0),
\qquad
\mathcal H=\operatorname{diag}(h,h).
\]
The classical small-Rollnik result implies that
\eqref{eq:small-Rollnik-condition} guarantees the existence and
completeness of the Schrödinger wave operators; see
\cite{Kato1966,RodnianskiSchlag2004}. Under the realification correspondence
established above, this result gives the real wave operators
\begin{equation}
\Omega_\pm^{\mathbb R}
:=
\operatorname*{s-lim}_{t\to\pm\infty}
\mathbb U(-t)\mathbb U_0(t).
\label{eq:real-Rollnik-wave-operators}
\end{equation}

The real wave operators are orthogonal:
\begin{equation}
(\Omega_\pm^{\mathbb R})^*\Omega_\pm^{\mathbb R}
=
\Omega_\pm^{\mathbb R}(\Omega_\pm^{\mathbb R})^*
=I. \qquad \quad\operatorname{Ran}\Omega_\pm^{\mathbb R}=\mathcal K
\label{eq:real-wave-orthogonality}
\end{equation}
which expresses asymptotic completeness in the present small-Rollnik
setting. They also satisfy
\begin{equation}
\Omega_\pm^{\mathbb R}\mathcal J
=
\mathcal J\Omega_\pm^{\mathbb R},
\qquad
\mathcal H\Omega_\pm^{\mathbb R}
=
\Omega_\pm^{\mathbb R}\mathcal H_0
\quad\text{on }D(\mathcal H_0).
\label{eq:real-Rollnik-wave-properties}
\end{equation}

The scattering operator compares the incoming and outgoing free
asymptotic representatives of the same interacting state. Let
\(\phi_{\mathrm{in}}\in\mathcal K\) be an incoming free state and set
\begin{equation}
\Phi=\Omega_-^{\mathbb R}\phi_{\mathrm{in}}.
\label{eq:incoming-interacting-state}
\end{equation}
By asymptotic completeness, there exists a unique outgoing free state
\(\phi_{\mathrm{out}}\in\mathcal K\) such that
\begin{equation}
\Phi=\Omega_+^{\mathbb R}\phi_{\mathrm{out}}.
\label{eq:outgoing-interacting-state}
\end{equation}
Therefore,
\[
\Omega_+^{\mathbb R}\phi_{\mathrm{out}}
=
\Omega_-^{\mathbb R}\phi_{\mathrm{in}}.
\]
Since \(\Omega_+^{\mathbb R}\) is orthogonal,
\((\Omega_+^{\mathbb R})^{-1}=(\Omega_+^{\mathbb R})^*\), and hence
\begin{equation}
\phi_{\mathrm{out}}
=
\mathbb S_{\mathrm R}\phi_{\mathrm{in}},
\qquad
\mathbb S_{\mathrm R}
:=
(\Omega_+^{\mathbb R})^*\Omega_-^{\mathbb R}.
\label{eq:real-scattering-operator}
\end{equation}
Thus \(\mathbb S_{\mathrm R}\) maps each incoming free asymptotic state
to the corresponding outgoing free asymptotic state.

The properties of the wave operators imply
\begin{align}
\mathbb S_{\mathrm R}^*\mathbb S_{\mathrm R}
&=
\mathbb S_{\mathrm R}\mathbb S_{\mathrm R}^*
=I,
\label{eq:real-scattering-orthogonality}\\
\mathbb S_{\mathrm R}\mathcal J
&=
\mathcal J\mathbb S_{\mathrm R},
\label{eq:scattering-complex-structure}\\
\mathcal H_0\mathbb S_{\mathrm R}
&=
\mathbb S_{\mathrm R}\mathcal H_0
\quad\text{on }D(\mathcal H_0).
\label{eq:scattering-free-energy}
\end{align}
Consequently,
\begin{align}
g(\mathbb S_{\mathrm R}\phi,\mathbb S_{\mathrm R}\xi)
&=g(\phi,\xi),
\label{eq:scattering-metric}\\
\omega(\mathbb S_{\mathrm R}\phi,\mathbb S_{\mathrm R}\xi)
&=\omega(\phi,\xi).
\label{eq:scattering-symplectic-form}
\end{align}
Thus the real scattering operator preserves the metric, symplectic form, complex structure, and free energy.

\subsection{Cross-section}

Since \(\mathbb S_{\mathrm R}\) commutes with \(\mathcal H_0\), scattering conserves the free energy. Therefore, in momentum space, \(\mathbb S_{\mathrm R}\) acts independently on each energy shell. We denote its action at fixed energy by \(\mathbb S_{\mathrm R}(E)\).

For the following stationary calculation, assume that \(V\) is a smooth, finite-range central potential,
\[
V(\mathbf x)=V(r),
\qquad
V(r)=0\quad\text{for }r>R
\]
for some \(R>0\). Under these standard assumptions, the usual stationary partial-wave expansion and large-distance scattering asymptotics apply.

Rotational invariance implies that different angular-momentum channels do not mix. At the fixed energy \(E=k^2/(2m)\), the action in the \(\ell\)-th partial wave is determined by the phase shift \(\delta_\ell(E)\). The magnetic quantum number is unaffected by the central potential. In the usual complex formulation, the corresponding partial-wave factor is
\begin{equation}
S_\ell(E)
=
e^{2i\delta_\ell(E)}.
\label{eq:complex-phase-shift}
\end{equation}
In the real Kähler formulation, the same transformation is represented by
\begin{equation}
\mathbb S_\ell(E)
=
e^{2\delta_\ell(E)\mathcal J}
=
\cos\bigl(2\delta_\ell(E)\bigr)I
+
\sin\bigl(2\delta_\ell(E)\bigr)\mathcal J.
\label{eq:real-phase-shift}
\end{equation}
Thus the usual complex phase shift becomes a rotation through the angle \(2\delta_\ell(E)\) in the real Kähler plane. The corresponding complex scattering amplitude is
\begin{equation}
f(\theta)
=
\frac{1}{2ik}
\sum_{\ell=0}^{\infty}
(2\ell+1)
\bigl(S_{\ell}(E)-1\bigr)
P_{\ell}(\cos\theta).
\label{eq:complex-scattering-amplitude}
\end{equation}
Writing $f(\theta)=f_{1}(\theta)+if_{2}(\theta)$, multiplication by
$f(\theta)$ is represented on the real Kähler plane by
\begin{equation}
\mathbb F(\theta)
=
f_{1}(\theta)I+f_{2}(\theta)\mathcal J
=
\begin{pmatrix}
f_{1}(\theta)&-f_{2}(\theta)\\
f_{2}(\theta)&f_{1}(\theta)
\end{pmatrix}.
\label{eq:real-scattering-amplitude}
\end{equation}

Let $\mathbf e=(1,0)^{T}$ denote a unit incident amplitude. The large-distance form of the stationary scattering state is
\begin{equation}
\phi_{E}(\mathbf{x})
\sim
e^{kz\mathcal J}\mathbf e
+
\frac{e^{kr\mathcal J}}{r}
\mathbb F(\theta)\mathbf e,
\qquad r\longrightarrow\infty.
\label{eq:real-scattering-asymptotics}
\end{equation}
The first term is the incident plane wave and the second is the outgoing spherical wave.

Using the probability current introduced in Sect.~7,
\begin{equation}
j_{i}(\phi)
=
\frac{1}{m}
\omega_{\mathbf{x}}
\left(
\phi,\frac{\partial\phi}{\partial x_{i}}
\right),
\end{equation}
the incident and scattered radial currents are
\begin{equation}
j_{\mathrm{in}}=\frac{k}{m},
\qquad
j_{\mathrm{sc},r}
=
\frac{k}{mr^{2}}
\left|\mathbb F(\theta)\mathbf e\right|^{2}
+o(r^{-2}).
\label{eq:real-scattering-currents}
\end{equation}
Moreover, since $\mathcal J^{T}=-\mathcal J$ and $\mathcal J^{2}=-I$,
\begin{equation}
\mathbb F(\theta)^{T}\mathbb F(\theta)
=
\bigl(f_{1}(\theta)^{2}+f_{2}(\theta)^{2}\bigr)I
=
|f(\theta)|^{2}I.
\label{eq:amplitude-norm-preservation}
\end{equation}
It follows that
\begin{equation}
\frac{d\sigma_{\mathbb R}}{d\Omega}
=
\lim_{r\to\infty}
\frac{r^{2}j_{\mathrm{sc},r}}{j_{\mathrm{in}}}
=
|f(\theta)|^{2}
=
\frac{d\sigma_{\mathbb C}}{d\Omega}.
\label{eq:cross-section-preservation}
\end{equation}
Hence the complex and real Kähler formulations have the same phase shifts and differential cross sections, while complex phase multiplication is represented by rotation in the real Kähler plane.

\section{Conclusion}

In this paper, we developed a framework for quantum dynamics, scattering, and tunneling in real quantum mechanics. By establishing a real analogue of Stone's theorem, we derived the general form of the evolution equation \eqref{GEE} on a real K\"ahler space and showed that the resulting evolution is both orthogonal and symplectic.

We then formulated scattering theory in real quantum mechanics in terms of real wave operators and the associated $S$-matrix. For the class of potentials considered, we showed that the scattering cross section equivalent with that obtained in standard complex quantum mechanics.

% {\bf OLD}: Finally, we analyzed quantum tunneling in the real formulation and proved that the tunneling probability is identical to its complex counterpart. This result demonstrates that real symplectic dynamics can reproduce squeezing into a classically forbidden region without the explicit use of complex-valued wave functions, by using symplectic dynamics. Though in real quantum mechanics dynamics is symplectic like in classical mechanics.

Finally, we analyzed quantum tunneling in the real formulation and proved that the tunneling probability is identical to its complex counterpart. This result shows that real quantum mechanics reproduces penetration into a classically forbidden region without introducing complex-valued wave functions explicitly. Instead, the relevant quantum dynamics is encoded in the symplectic structure of the real Kähler space. Although the evolution in real quantum mechanics is symplectic, as in classical Hamiltonian mechanics, the underlying state space and dynamical structure remain genuinely quantum.

It would also be interesting to investigate another characteristic quantum phenomenon, namely interference, within the real formulation of quantum mechanics. A further important direction is the semiclassical expansion of real quantum mechanics, since its quantum evolution is itself described by a symplectic structure closely analogous to classical Hamiltonian dynamics.

\section*{Acknowledgement}
We would like to thank I.Ya.Aref’eva for  useful discussions. We also would like to thank A.S.Trushechkin,  A.E.Teretenkov and R.Singh for useful discussions at Quantum Mathematical Physics Seminar at Steklov Mathematical Institute.

% It would be interesting to consider another quantum mechanical phenomena, interference, which is typical for complex formulation of quantum mechanics. It would be to investigate semi-classical expansion of real quantum mechanics since in this case we have classical Hamiltonian evolution even in the quantum level.    

\newpage

\appendix
\section{Appendix: Quantum dynamical system}

One can ask a question what a general form of quantum and classical dynamics. We consider a Schrodnger equation of the form 
\begin{equation}
    i\dot\psi=H \psi,
\end{equation}
where H is a self-adjoint operator in a separable Hilbert space $\mathcal{H}$. One can show that any such dynamical system is unitary equivalent to a system of classical harmonic oscillators with the following Hamiltonian
\begin{equation}
\label{Kam}
K=\frac{1}{2}\int_X\lambda(x)(p^2(x)+q^2(x))d\mu(x),
\end{equation}
here $X$ is a mesurable space with a measure $\mu$. This Hamiltonian will generate classical equation of motion
\begin{equation}
\label{Hameq}
    \dot q(x)=\lambda(x)p(x) \qquad
\dot{p}(x)= -\lambda(x)\, q(x),
\end{equation}
which gives
\begin{equation}
\ddot{q}(x) + \lambda(x)^2 q(x) = 0.
\end{equation}
This given in \cite{IVpadic}. 
Other classical dynamical system which defined as a triplet  $(M,\nu,\alpha_t)$ where $M$ is a measurable space, $\nu$ is a measure on $M$, and 
$\alpha_t : M \to M$ is a one-parameter group of measure-preserving transformations. It can be show that this system are equivialent to the classical form \eqref{Kam}.

The Schrodinger equation defined unitary evolution operator $U_t$. If we have a quantum dynamical system $(\mathcal{H},U_t)$, then according to the spectral theorem there exist measurable space $X(\mu)$ and unitary mapping $W: \mathcal{H}\to L^2(d_\mu)$, where $L^2(d_\mu)$ is a Hilbert space of complex numbers integral functions on $X$.
\begin{equation}
W U_t W^{*}  = e^{-it\lambda}, 
\end{equation}
here $\lambda(x)$ is a measurable real valued function on $X$.
One has
\begin{equation}
    e^{-it\lambda}\phi(x)= e^{-it\lambda(x)}\phi(x).
\end{equation}
Let \eqref{Kam} be generator of the unitary group $U_t$, 
\begin{equation}
 U_t=e^{-itH}.
\end{equation}
\begin{equation}
    i\dot\psi=H\psi,
\end{equation} 
where $\psi\in \mathcal{H}$. After unitary transformation $W$ this equaiton become 
\begin{equation}
     i\dot\phi(x)=\lambda(x)\phi(x),
\end{equation}
where $\phi\in L^2$.
Now if we set, $\phi(x)=q(x)+ip(x)$, then the Schrodinger equation get the form of classical equation for harmonic oscillators 
\begin{equation}
     i\dot\phi(x)=\lambda(x)(q(x)+ip(x)),
\end{equation}
\begin{equation}
     i\dot q(x)-\dot p(x)=\lambda(x)(q(x)+ip(x))
\end{equation}
It equivalent to the systems of complex and real parts, we see 
\begin{equation}
    \dot q(x)=\lambda(x)p(x) \qquad\dot p(x)=-\lambda(x)q(x)
\end{equation}
Therefore we obtain classical equation of motion corresponding to the classical Hamiltonian $K$ \eqref{Kam}.
\subsection{Quantum Harmonic oscillator}
Let us consider quantum harmonic oscillator with the Hamiltonian in $L^2(\mathbb{R})$ in the form
\begin{equation}
    H=-\frac{^2}{2m}\frac{d^2}{dx^2}+\frac{1}{2}m\omega^2 x^2
\end{equation}
where eigenfunction
\begin{equation}
\psi_n(x)
=
\left(\frac{m\omega}{\pi}\right)^{1/4}
\frac{1}{\sqrt{2^n n!}}\,
H_n\!\left(\sqrt{\frac{m\omega}{}}\,x\right)
\exp\!\left(-\frac{m\omega x^2}{2}\right),
\qquad n=0,1,2,\dots
\end{equation}
With eigenvalues 
\begin{equation}
    E_n=\omega (n+\frac{1}{2}).
\end{equation}
We represent the solution of the Hamilton equation \eqref{Hameq} in the form 
\begin{equation}
    q(x,t)=\sum_{n=0}^\infty q_n(t)\psi_n(x), \qquad p(x,t)=\sum_{n=0}^\infty p_n(t)\psi_n(x),
\end{equation}
where $q_n(t)$ and $p_n(t)$ satisfy the equation
\begin{equation}
  \dot q_n(t) =^{-1}E_n p_n(t), \qquad \dot p_n(t) = -^{-1}E_n q_n(t), \qquad n=0,1, ...
\end{equation}
So we obtain the systems of harmonic oscillator with the Hamiltonian \begin{equation}
H_{osc}= \frac{1}{2} \sum_{n=0}^\infty E_n(q_n^2+p_n^2).   
\end{equation}
\subsection{Continuous spectrum}
In this subsection we consider an example of integrability of quantum dynamics by using the wave operators.
Let $H$ and $H_0$ be self-adjoint operators on Hilbert space,
\begin{equation}
    H=H_0+V
\end{equation}and
exist the limits
\begin{equation}
\lim_{t \to \pm\infty} e^{itH} e^{-itH_0} = \Omega_{\pm},
\end{equation}
which are called the wave operators. One considers only absolute continuous parts
of the Hamiltonians. The wave operators $\Omega_{\pm}$ obey the intertwining property
$H\Omega_{\pm} = \Omega_{\pm}H_0$. If $H_0$ is a simple ``free'' operator then the existence of
the intertwining relation can be used to claim the integrability of the system described
by the more complicated Hamiltonian $H$. Even if it is so, still there is a problem of
how to compute explicitly the wave operators or the scattering operator ($S$-matrix)
\begin{equation}
S = \Omega_{+}^{*}\Omega_{-}
\end{equation}
which is a task of primary interest.

As an example $L^2( \mathbb{R}^N)$ we consider free hamiltonian
\begin{equation}
    H_0=-\frac{\nabla}{2m}
\end{equation}
and the potential $V(x)$ which vanishing for large $x$. In this case $\Omega_{\pm}$ are unitary operator.

\section{Appendix: Classical dynamical system}
We consider the classical Hamiltonian system in a phase space $\mathbb{R}^{2N}$ with the Hamiltonian 
 \begin{equation}
  H_{cl}= H_{cl}(p,q)
\end{equation}
with canonical equation of motion
\begin{equation}
    \dot q^j=\frac{\partial H_{cl}}{\partial p_j},\qquad
\dot{p}^j= -\frac{\partial H_{cl}}{\partial q_j}, \qquad j=1,...,n.
\end{equation}
This Hamiltonian defines classical dynamical system ($\mathbb{R}^{2N},\nu,\alpha_t$), where $\nu$ is a Lebeg measure $d\nu=dpdq$ with invariant under the Hamiltonian flow $\alpha_t$.

More generally let us consider classical dynamical system ($M,\nu,\alpha_t$), where pair ($M,\nu$) is a measurable space at $\alpha_t$ is a group of automorphism preserved the measure $\nu$. Let us consider the Hilbert space of complex valued functions $L^2(M,\mu)$ and the unitary operator $U_t$ defined by $U_tf(x)=f(\alpha_t,x)$. By the spectral theorem there exist a unitary mapping $W: L^2(M,\nu)\to L^2(X,\mu)$ is a measurable Hilbert space 
\begin{equation}
W U_t W^{*}  = e^{-it\lambda}.
\end{equation}
We can repeat now the application of the spectral theorem as it was done in quantum case and obtained that any classical dynamical system is equivalent to the system of classical Harmonic oscillators with the Hamiltonian $K$ \eqref{Kam}.

\end{document}